\documentclass{sigplanconf}
\usepackage{times,amsmath,epsfig}

\usepackage[]{graphics}
\usepackage[]{graphicx}

\usepackage{latexsym}
\usepackage{amssymb,amsmath,amsfonts}
\usepackage{stmaryrd}
\usepackage{newalg}

\usepackage{tikz}
\usetikzlibrary{arrows}
\usetikzlibrary{shapes.misc}

\usepackage{theorems}

\newcommand{\sema}[2]{\llbracket #1 \rrbracket
}

\newcommand{\staterule}[3]{%
  $\begin{array}{@{}l}%
   \mbox{#1}\\%
   \begin{array}{l}
   #2\\
   \hline
   \raisebox{0ex}[2.5ex]{\strut}#3%
   \end{array}
  \end{array}$}

\newlength{\mygap}
\newlength{\smallgap}
\newcommand{\hbra}{
\hbox to \linewidth{\vrule width0.3mm height 1.8mm depth-0.3mm
                    \leaders\hrule height1.8mm depth-1.5mm\hfill
                    \vrule width0.3mm height 1.8mm depth-0.3mm}}
\newcommand{\hket}{
\hbox to \linewidth{\vrule width0.3mm height1.5mm
                    \leaders\hrule height0.3mm\hfill
                    \vrule width0.3mm height1.5mm}}

\newlength{\entryindent}
\newlength{\reductionind}
\newlength{\clausetextind}
\newlength{\entrytextind}
\newlength{\labelind}
\newlength{\entrytext}
\newlength{\clausetext}

\newlength{\tabone}
\newlength{\tabtwo}
\newlength{\tabthree}
\newlength{\tabfour}
\newlength{\tabfive}

\newcommand{\ratio}{.18}

\newcounter{displayCounter}[section]

\newcommand{\refDisplayHere}[1]{\thesection.\ref{#1}}

\newcommand{\labelDisplay}[1]{\refstepcounter{displayCounter}
                           \label{#1}
                           \addtocounter{displayCounter}{-1}}

\newenvironment{display}[1]{
  \refstepcounter{displayCounter}
  \setlength{\entrytext}{\linewidth}
  \addtolength{\entrytext}{-\ratio\linewidth}
  \addtolength{\entrytext}{-3em}
  \setlength{\clausetext}{\entrytext}
  \addtolength{\clausetext}{1.5em}
\begin{tabbing}
  \hspace{1.5em}\=\hspace{\ratio\linewidth}\=\hspace{1.5em}\= \kill
  {Table \thesection.\thedisplayCounter. \emph{#1}}\\[-.8ex]
  \hbra\\[-.8ex]
  }{\\[-.8ex]\hket
  \end{tabbing}}

\newcommand{\iftimes}[1]{}
\newcommand{\ifshort}[1]{}

\newcommand{\ifnoNP}[1]{}

\newcommand{\ileave}{\cdot}

\newcommand{\sym}[1]{sym(#1)}

\newcommand{\ssem}[1]{\sema{#1}{}}

\newcommand{\Nat}{\mathbb{N}}

\newcommand{\edge}[3]{(#1, #2, #3)}

\newcommand{\fin}[1]{in(#1)}
\newcommand{\fout}[1]{out(#1)}
\newcommand{\econc}[1]{\pi (#1)}

\newcommand{\elem}[1]{#1}
\newcommand{\eleme}{\elem{e}}
\newcommand{\elemei}{\eleme_{i}}
\newcommand{\elemej}{\eleme_{j}}
\newcommand{\xeleme}[1]{\elem{e}_{#1}}
\newcommand{\gelem}[2]{(#1,#2)}
\newcommand{\gin}[1]{#1.in}
\newcommand{\gout}[1]{#1.out}
\newcommand{\geleme}{(\gin{\eleme}, \gout{\eleme})}

\newcommand{\gschema}[1]{\mathcal{#1}}
\newcommand{\schemas}{\gschema{S}}

\newcommand{\rsem}[1]{L(#1)}
\newcommand{\ls}[2]{#1}
\newcommand{\gsem}[2]{\ssem{#1}_{#2}}
\newcommand{\ggsem}[1]{\gsem{#1}{G}}
\newcommand{\rconc}{\circ}
\newcommand{\back}[1]{#1^{-}}

\newcommand{\xcount}[3]{#1^{#2,#3}}
\newcommand{\xset}[2]{\{ #1 \}_{#2}}
\newcommand{\pset}[1]{\mathcal{#1}}
\newcommand{\psete}{\pset{E}}
\newcommand{\psetei}{\psete_{i}}
\newcommand{\judbasic}[3]{\vdash_{#1} {#2} : {#3}}
\newcommand{\first}[1]{first(#1)}

\newcommand{\rpq}{RPQ}
\newcommand{\nre}{NRE}
\newcommand{\gxp}{GXP}
\newcommand{\paths}[1]{Paths(#1)}

\newcommand{\reps}[1]{{#1}^{*}}
\newcommand{\plus}[1]{{#1}^{+}}
\newcommand{\xsat}[1]{SAT(#1)}
\newcommand{\rsat}{\xsat{\rpq}}
\newcommand{\xbra}{]}
\newcommand{\xnorm}[1]{norm(#1)}

\newcommand{\sdnf}{DNF}

\newif{\ifMarginalComments}
 \MarginalCommentstrue   

\newcommand{\DC}[1]{}

  \newcommand{\CS}[1]{}
\begin{document}

\setlength{\pdfpageheight}{\paperheight}
\setlength{\pdfpagewidth}{\paperwidth}

\conferenceinfo{CONF 'yy}{Month d--d, 20yy, City, ST, Country} 
\copyrightyear{20yy} 
\copyrightdata{978-1-nnnn-nnnn-n/yy/mm} 
\doi{nnnnnnn.nnnnnnn}



\title{Typing Regular Path Query Languages for Data Graphs}


\authorinfo{Dario Colazzo}
           {LAMSADE - Universit\'e Paris-Dauphine}
           {dario.colazzo@dauphine.fr}
\authorinfo{Carlo Sartiani}
           {DIMIE - Universit\`a della Basilicata}
           {sartiani@gmail.com}

\maketitle
%

\begin{abstract}
Regular path query languages for data graphs are essentially \emph{untyped}. The lack of type information greatly limits the optimization opportunities  for query engines and makes application development more complex. In this paper we discuss a simple, yet expressive, schema language for edge-labelled data graphs.  This schema language is, then, used to define a query type inference approach with good precision properties.
\end{abstract}





\section{Introduction}\label{sec:intro}

In the last few years graph databases gained more and more relevance in application areas such as the Semantic Web, social networks, bioinformatics, network traffic analysis, and crime detection. This led to the definition of many query formalisms for graph databases, like, for instance, \emph{regular path queries} (RPQs \cite{DBLP:journals/siamcomp/MendelzonW95}), \emph{nested regular expressions} (NREs \cite{DBLP:journals/ws/PerezAG10}), \emph{conjunctive regular path queries} (CRPQs \cite{DBLP:conf/pods/ConsensM90}), GXPath \cite{DBLP:conf/icdt/LibkinMV13}, and their derivatives. All these languages are based on the idea of specifying regular expressions describing paths in the input graph, and can be considered, to some extent, a generalization of existing path query languages for semistructured data (see XPath \cite{XPath2.0}, for instance). Regular path query languages are often used in other graph query languages, like Cypher \cite{cypher} or PQL \cite{pql}, to specify patterns in variable binding clauses.

Regular path query languages are essentially \emph{untyped}. This means that one cannot statically infer the structure of query results (\emph{type inference}), check if the results satisfy a given schema (\emph{type-checking}), and verify if the query results would always be empty (\emph{query correctness}). Furthermore, the lack of type information greatly limits the optimization opportunities  for query engines and makes application development more complex.

\paragraph*{Our Contribution}\label{par:ourcontrib}

In this paper we describe   a simple, yet expressive, schema language for edge-labelled data graphs. A schema is formed by a collection of \emph{schema elements}, each one describing the set of incoming and outgoing edges of a class of graph nodes; edges are specified through regular expressions. Unlike what happens in other schema languages for graphs \cite{DBLP:conf/sigmod/ShaoWL13} \cite{DBLP:conf/icdt/StaworkoBGHPS15}, that  allow the designer to describe in full detail the structure of outgoing edges as well as the structure of node values, but give her very limited modelling choices for incoming edges, our schema language makes no distinction between incoming and outgoing edges, and gives the designer the same modeling tools for both classes of edges.


This increased expressive power has the drawback that, as we will show in Section \ref{sec:preldefs}, in the general case, schema emptiness checking is undecidable; hence, a few restrictions on regular expressions describing edges are needed in order to ensure that the semantics of graph schemas is well defined, and to make schema emptiness decidable. The resulting class of schemas is named  \emph{well-formed schemas}  and can be  viewed as a generalization of DTDs \cite{DBLP:journals/tods/MartensNSB06} to data graphs. The proposed language is a first step towards the definition and analysis of even more powerful schema languages for data graphs.

In the second part of the paper we leverage on well-formed schemas  to build a type inference system, working in polynomial time,  for RPQs, NREs, and GXPath queries with good soundness and completeness properties; in particular, this type inference system is sound and complete on RPQs,  while completeness has to be relaxed on NREs and GXPath queries. This means that, by using this system, it is possible to decide whether an RPQ is satisfiable on graphs conforming to a given schema in polynomial time.

\paragraph*{Paper Outline}\label{par:papout}

The paper is structured as follows. In Section \ref{sec:preldefs} we first describe the data model and the type language used in our approach; then, we present a schema language for data graphs and discuss the emptiness problem for the resulting schemas. In Section \ref{sec:querylan}, next, we survey regular path query languages and describe their semantics. In Section \ref{sec:rules}, then, we present our type inference systems. In Sections \ref{sec:relworks} and \ref{sec:concl}, we discuss some related works and  draw our conclusions.


\section{Preliminary Definitions}\label{sec:preldefs}

\subsection{Data Model and Type Language}\label{subsec:datamodel}

Following \cite{DBLP:conf/icdt/LibkinMV13}, we model a \emph{data graph} as an edge-labelled graph, as shown below.

\begin{definition}[Data Graph]
Given a finite alphabet $\Sigma$ and a (possibly) infinite value domain $\mathcal{D}$, a data graph $G$ over $\Sigma$ and $\mathcal{D}$ is a triple $G = (V, E, \rho)$, where:

\begin{itemize}
\item $V$ is a finite set of nodes;

\item $E \subseteq V \times \Sigma \times V$ is a set of labelled, directed edges $\edge{v_{i}}{a}{v_{j}}$;

\item $\rho : V \rightarrow \mathcal{D}$ is a mapping from nodes to values. 
\end{itemize}

\end{definition}

Given a node $v$, we  indicate with $\fin{v}$ and $\fout{v}$ the set of incoming and outgoing edges, respectively. Formally: 
{
\begin{itemize}
\item $\fin{v} = \{\edge{v^{\prime}}{a}{v} \in E \mid v^{\prime} \in V \wedge a \in \Sigma \}$;

\item $\fout{v} = \{\edge{v}{a}{v^{\prime}} \in E \mid v^{\prime} \in V \wedge a \in \Sigma \}$.
\end{itemize}}

 We assume that sequences of   outgoing (incoming) edges of a node are unordered, as it is often the case in graph databases.
 Given a set of edges $S_{E} \subseteq E$, we will indicate with $\econc{S_{E}}$ the unordered concatenation of the labels of the edges in $S_{E}$. 

This data model is general enough to capture many practical use graphs, ranging from RDF data to social network graphs, as shown by the following example.

\begin{example}\label{ex:first}
Consider the graph shown in Figure \ref{fig:first}. This graph contains bibliographic information coming from a fragment of the RDF representation of the DBLP repository \cite{DBLPRDF}. As in \cite{DBLP:conf/icdt/LibkinMV13}, we indicate the value of a node inside its graphical representation, and use RDF properties to label edges.
\end{example}
\begin{figure}
{\tiny{
\begin{tikzpicture}[->,>=stealth',shorten >=1pt,auto,node distance=2.5cm,
                    thick,main node/.style={rounded rectangle,draw,font=\sffamily\tiny}]

  \node[main node] (1) at (2,0) {jacm};
  \node[main node] (2) at (4,0) {HopcroftT74};
  \node[main node] (3) at (7,0) {Robert Endre Tarjan};
  \node[main node] (4) at (0,-2) {focs};
  \node[main node] (5) at (2,-2) {FOCS8};
  \node[main node] (6) at (4,-2) {HopcroftU67a};
  \node[main node] (7) at (7,-2) {John E. Hopcroft};
  \node[main node] (8) at (7,-4) {Jeffrey D. Ullman};

  \path[every node/.style={font=\sffamily\small}]
    (2) edge node  {journal} (1)
        edge node  {creator} (3)
        edge node  {creator} (7)
    (5) edge node  {series} (4)
    (6) edge node {partOf} (5)
         edge node {creator} (7)
         edge node {creator} (8);
\end{tikzpicture}}}
\caption{An RDF graph.}
\label{fig:first}
\end{figure}
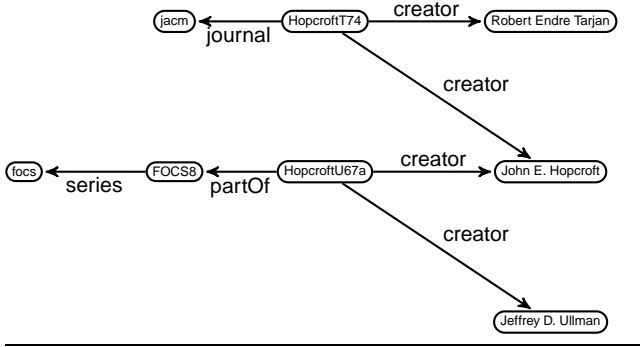

In this work we propose a schema language for data graphs that associates to each schema element a pair of regular expressions describing sequences of labels of the incoming and outgoing edges of each node. Regular expressions obey the following grammar:
\[
\begin{array}{lllll}
T  ::= &    \epsilon 
	  \ \mid\ a 
	  \ \mid\ T + T 
	  \ \mid\ T \ileave T
	  \ \mid\ \reps{T}
\end{array}
\]
where $\epsilon$ denotes the empty sequence, $a$ is a symbol in $\Sigma$, $+$ and  $\ileave$  denote, respectively, union and unordered concatenation, and $*$ is the Kleene star. As expected, unordered concatenation $\ileave$ is commutative, associative and has $\epsilon$ as neutral element. In particular, the expression $T_{1} \ileave \ldots \ileave T_{n}$ is equivalent to all  of its possible permutations. In the following we will also use $\plus{T}$ and $T?$ as abbreviations for $\reps{T} \cdot T$ and $T + \epsilon$.

The semantics of regular expressions  is denoted as $\rsem{-}$, denoting the minimal function  satisfying the following equations: 
\[
\begin{array}{llllllll}
\rsem{\epsilon} &  = & \{\epsilon\} \\ 
\rsem{a} & = & \{a\} \\
\rsem{T_{1} + T_{2}} &= & \rsem{T_{1}} \cup \rsem{T_{2}} \\
\rsem{T_{1} \ileave T_{2}} & = & \rsem{T_{1}}\ileave\rsem{T_{2}}\\ 
\rsem{\reps{T}} & = & \bigcup_{i\geq 0} \rsem{T}^{i} 
\end{array}
\]
where $L_1 \ileave L_2$ denotes  unordered language concatenation and is defined in the obvious way, while  for any $i \in \Nat$, $L^{i}=L \ileave L^{i-1}$ with   $L^{0}=\{\epsilon\}$.

\subsection{Schema Language}\label{subsec:schema}

Regular expressions are the building blocks of our schema language.

%
%
%

\begin{definition}
Given a regular expression $T$ over $\Sigma$, $\sym{T}$ is the set of symbols in $\Sigma$ appearing in $T$.
\end{definition}

\begin{definition}[Graph Schema Element]
Given a finite alphabet $\Sigma$, a schema element $\eleme$ over $\Sigma$ is a pair $\geleme$, where $\gin{\eleme}$ and $\gout{\eleme}$ are regular expressions over $\Sigma$. 
\end{definition}
The semantics of a schema element $\eleme$ is defined as follows. 
\begin{eqnarray*}
\ssem{\eleme} & = & \{v \mid \econc{\fin{v}} \in \rsem{\gin{\eleme}} \wedge \econc{\fout{v}} \in \rsem{\gout{\eleme}} \}
\end{eqnarray*}

\noindent A schema element, then, specifies constraints on the incoming and outgoing edges of a node. Consider, for instance, the following schema element:
$$e = \gelem{\ls{a}{1} \ileave \ls{b}{1} \ileave (\ls{c}{1} + \ls{d}{1})}{\ls{e}{2} \ileave \reps{\ls{h}{1}}}$$


This element describes graph nodes having an incoming $a$-edge, an incoming $b$-edge, as well as an incoming edge labelled with $c$ or $d$; these nodes must also have an outgoing $e$-edge together with zero or more  outgoing $h$-edges. 

In our schema language, hence, we not only impose constraints on outgoing edges, but also on incoming edges. This is in contrast to what happens in schema languages for XML data (e.g., DTDs \cite{XML1.1} and XML Schema \cite{XMLSchema-part1}). This choice is motivated by the observation that in a graph each vertex may have multiple incoming edges and, hence, multiple \emph{fathers}, while in an XML tree each node, except for the root, has a single father. Therefore, it is important to give the schema designer the ability to model the set of incoming edges, so to avoid potentially dangerous situations. Consider, for instance, a data graph describing a bibliographic database, where nodes can represent books, papers, authors, and publishers; of course, while author nodes can have incoming edges labelled with ``writtenBy'', they cannot allow for incoming edges with label ``publishedBy'', which, instead, are allowed for publisher nodes only.  

The use of regular expressions for modeling incoming edges makes our language quite different from existing graph schema languages like TSL \cite{DBLP:conf/sigmod/ShaoWL13} and SheX \cite{DBLP:conf/icdt/StaworkoBGHPS15}. In all these languages, the designer can use regular expressions to specify the sequence of outgoing edges for each node type; each edge is described by a label and by the type of the receiving node. Therefore, in these languages it is not possible to specify, for instance, that a node of a given type can have exactly one incoming edge of a given kind.

\begin{definition}[Graph Schemas]\label{def:grsch}
A graph schema $\schemas$  is  a  finite set of schema elements $\{\elemei\}_{i = 0}^{n}$  such that:
\begin{enumerate}
\item $\forall i \in [0..n]. \forall l \in \sym{\gin{\elemei}}. \exists j \in [0..n]. l \in \sym{\gout{\elemej}}$;

\item $\forall i \in [0..n]. \forall l \in \sym{\gout{\elemei}}. \ \exists j \in [0..n]. l \in \sym{\gin{\elemej}}$;

\item $\forall i,j \in [0..n]: (\gin{\elemei} \cap \gin{\elemej} = \emptyset \vee \gout{\elemei} \cap \gout{\elemej} = \emptyset)$.
\end{enumerate}
\end{definition}

Conditions 1 and 2 above are necessary to ensure that the schema cannot define graphs with dangling edges: any symbol used in an outgoing edge must also be used to label an incoming edge, and vice versa. As we will see later, these conditions are not sufficient to imply non-emptiness. Condition 3 guarantees the uniqueness of node typing: a graph node can be typed by at most one schema element.

Schema semantics is defined as follows. 
\begin{definition}[Graph Schema Semantics]\label{gsem}
A data graph $G = (V, E, \rho)$ over $\Sigma$ and $\mathcal{D}$ is described by a graph schema $\schemas$ ($G \in \ssem{\schemas}$) if and only for each $v \in V$ there exists $\elemei \in \schemas$ such that $v \in \ssem{\elemei}$.
\end{definition}

\begin{example}\label{ex:second}
Consider again the graph of Example \ref{ex:first}. This graph can be typed by the schema $\schemas  =  \{ \xeleme{1}, \xeleme{2}, \xeleme{3}, \xeleme{4}, \xeleme{5} \}$, where:
\[
\begin{array}{lllllllllllllll}
\xeleme{1} & = & \gelem{\epsilon}{(\ls{journal}{1} + \ls{partOf}{1}) \ \ileave \ \plus{(\ls{creator}{1})}}\\
\xeleme{2} & = & \gelem{\reps{journal}}{\epsilon}\\
\xeleme{3} & = & \gelem{\reps{partOf}}{\ls{series}{1}}\\
\xeleme{4} & = & \gelem{\reps{series}}{\epsilon}\\
\xeleme{5} & = & \gelem{\reps{creator}}{\epsilon}
\end{array}
\]
 \end{example}
 
 \subsection{Schema Emptiness}\label{subsec:emptiness}

A graph schema, even though it satisfies all the properties of Definition \ref{def:grsch}, may be empty, and it could be difficult for the user to figure out whether the schema she has defined is empty.  For a simple schema like the following one, emptiness can be easily detected. 

\begin{example}\label{ex:empty}
Consider the graph schema $\schemas = \{\xeleme{1}, \xeleme{2} \}$, where:
\[
\begin{array}{llllllllll}
\xeleme{1} & = & (\epsilon, \ls{a}{1} \ileave \ls{b}{1} \ileave \ls{c}{1} \ileave \ls{c}{1} ) &\\
\xeleme{2} & = & (\ls{a}{1} \ileave \ls{b}{1} \ileave \ls{c}{1}, \epsilon)
\end{array}
\]

This schema satisfies conditions 1-3 of Definition \ref{def:grsch}. However, it is empty as cardinality constraints expressed by regular expressions of incoming and outgoing edges are incompatible. 
\end{example}

For some schemas, checking compatibility between incoming and outgoing edges can be far from being obvious, as happens for the following one.

\begin{example}\label{ex:nonempty}
Consider the graph schema $\schemas = \{\xeleme{1}, \xeleme{2},\xeleme{3} \}$, where:
\[
\begin{array}{llllllllll}
\xeleme{1} & = & (\epsilon, \ls{a}{1} \ileave \ls{b}{1} \ileave \ls{c}{1} \ileave \ls{c}{1}  \ileave \ls{c}{1} \ileave \ls{c}{1}) \\ 
\xeleme{2} & = & (\ls{a}{1} \ileave \ls{b}{1} \ileave \ls{c}{1}, \epsilon)\\
\xeleme{3} & = & (\ls{c}{1} \ileave \ls{c}{1}, \epsilon)
\end{array}
\]

In this schema each $\xeleme{1}$ node produces 4 outgoing c-edges, that are consumed by $\xeleme{2}$ and $\xeleme{3}$ nodes. This schema is not empty, as it possible to build a well-formed graph comprising 2 $\xeleme{1}$ nodes, 2 $\xeleme{2}$ nodes, and 3 $\xeleme{3}$ nodes.
\end{example}

Without imposing restrictions on the class of regular expressions being used, checking the emptiness of a schema is not decidable. To show this undecidability result, it is necessary to establish an equivalence between graph schemas and homogeneous systems of linear diophantine equations with parameters. Indeed, we associate to each schema element a distinct variable, and build, for each symbol, a polynomial equation  describing the produced and consumed edges labelled with that symbol. Each symbol equation contains the variables of the schema elements producing or consuming edges labelled with that symbol; the coefficient of each variable describes the number of produced or consumed edges. The result is an homogeneous system which has a non-zero natural solution if and only if  the schema is not empty. The following example   illustrates this approach. 

\begin{example}\label{ex:sysex}
Consider again the schema of Example \ref{ex:empty}. This empty schema consists of two schema elements ($\xeleme{1}$ and $\xeleme{2}$) to which we can associate variables $x$ and $y$. Regular expressions in the schema use three different  symbols ($\ls{a}{1}$, $\ls{b}{1}$, and $\ls{c}{1}$), so we have to define the following three linear equations:
\[
\begin{array}{lllllllll}
\ls{a}{1}. & x - y = 0 \\
\ls{b}{1}. & x - y = 0 \\
\ls{c}{1}. & 2x - y = 0
\end{array}
\]

In the first equation variable $x$ has coefficient $1$, as $\xeleme{1}$ produces an $a$-edges, while variable $y$ has coefficient $-1$ since $\xeleme{2}$ consumes an $a$-edge. As it can be easily seen, the only solution of this system is (0,0,0).

\medskip

Consider now the schema of Example \ref{ex:nonempty}. As illustrated before, this schema is not empty and comprises three schema elements ($\xeleme{1}$, $\xeleme{2}$, and $\xeleme{3}$) to which we can associate variables $x$, $y$, and $z$. As for the previous example, we have three distinct symbols in the schema, so we can define a system with the following linear equations:
\[
\begin{array}{lllllllll}
\ls{a}{1}. & x - y = 0 \\
\ls{b}{1}. & x - y = 0 \\
\ls{c}{1}. & 4x - y -2z = 0
\end{array}
\]

It easy to see that (2,2,3) is a solution for this system. This means that it is possible to build a graph with 2 $\xeleme{1}$ vertices, 2 $\xeleme{2}$ vertices, and 3 $\xeleme{3}$ vertices.
\end{example}

In the case a schema contains Kleene stars, it is possible to build an equivalent diophantine system by introducing natural parameters, as shown in the following example.

\begin{example}\label{ex:exstar}
Consider the graph schemas $\schemas = \{\xeleme{1}, \xeleme{2},\xeleme{3} \}$, where:
\[
\begin{array}{llllllllll}
\xeleme{1} & = & (\epsilon, \ls{a}{1} \ileave \ls{b}{1} \ileave \reps{(\ls{c}{1} \ileave \ls{c}{1}  \ileave \ls{c}{1} \ileave \ls{c}{1})}) \\ 
\xeleme{2} & = & (\reps{(\ls{a}{1} \ileave \ls{b}{1} \ileave \ls{c}{1})}, \epsilon)\\
\xeleme{3} & = & (\ls{c}{1} \ileave \ls{c}{1}, \epsilon)
\end{array}
\]

To build an equivalent  system we can associate a distinct parameter to each occurrence of the Kleene star; in particular, we associate the parameter $h_{1}$ to the occurrence in $\xeleme{1}$, and a parameter $h_{2}$ to the occurrence in $\xeleme{2}$. The resulting system is the following:
\[
\begin{array}{lllllllll}
\ls{a}{1}. & x - h_{2}y = 0 \\
\ls{b}{1}. & x - h_{2}y = 0 \\
\ls{c}{1}. & 4h_{1}x - h_{2}y -2z = 0
\end{array}
\]

While this system contains equations that are linear in variables $x$, $y$, and $z$, coefficients are no longer constant and can assume any value in $\mathbb{N}$.

\end{example}

In the case of the schemas of  Example \ref{ex:sysex}, it is quite easy to verify if the corresponding system is  consistent and has a non-trivial, positive integer solution. Indeed, as pointed out in \cite{DBLP:conf/mfcs/Domenjoud91}, it suffices to build the \emph{convex hull} of the set of m-dimensional points defined by the columns of the system coefficient matrix and to check if $\vec{0}$ is contained in this polytope.

However, in the case of the schema of Example \ref{ex:exstar}, this approach can no longer be used. Indeed, the coefficient matrix contains parameters that prevent one from computing the convex hull. The problem of the consistency of homogeneous systems of diophantine equations with parameters has been already studied \cite{DBLP:conf/icalp/XieDI03,clauss98}. In \cite{DBLP:conf/icalp/XieDI03} Xie et al. proved that, even if we restrict to linear polynomial of parameters (no nested Kleene stars), there exists a fixed $k > 2$ such that the problem is undecidable if the system contains at least $k$ equations (i.e., the schema uses at least $k$ distinct symbols), and that the problem is decidable for systems of 2 equations; in \cite{clauss98} Clauss showed that the problem is decidable if the system contains a single parameter, two variables, and any number of equations.

These results motivate the need for a restriction on the class of schemas that ensures the non-emptiness of the schema. To develop such a restriction, we propose here an approach based on   several ingredients. The first one consists of restricting the kind of regular expressions that can be used in element types. As seen before, one source of difficulty is the presence of regular expressions with multiple occurrences of a  symbol. Another aspect that complicates the problem is nesting of repetitions: indeed, it is known that the consistency of systems of Diophantine equations is undecidable if equation degree is greater or equal to 4 \cite{opac-b1083621}, and nested Kleene stars in a schema just increase the degree of the equations in the corresponding system. Consider, for instance, the following schema:
\[
\begin{array}{llllllllll}
\xeleme{1} & = & (\epsilon, \reps{(\ls{a}{1} \ileave \reps{(\ls{b}{1} \ileave \reps{(\ls{c}{1} \ileave \ls{c}{1}  \ileave \ls{c}{1} \ileave \ls{c}{1})})})}) \\ 
\xeleme{2} & = & (\reps{(\ls{a}{1} \ileave \ls{b}{1} \ileave \ls{c}{1})}, \epsilon)\\
\xeleme{3} & = & (\ls{c}{1} \ileave \ls{c}{1}, \epsilon)
\end{array}
\]
The corresponding system, which uses four parameters $h_{1}$, $h_{2}$, $h_{3}$, and $h_{4}$,   has degree 4, as shown below:
\[
\begin{array}{lllllllll}
\ls{a}{1}. & h_{1}x - h_{2}y = 0 \\
\ls{b}{1}. & h_{1}h_{3}x - h_{2}y = 0 \\
\ls{c}{1}. & 4h_{1}h_{3}h_{4}x - h_{2}y -2z = 0
\end{array}
\]

Inspired by our previous works \cite{DBLP:conf/dbpl/GhelliCS07, DBLP:journals/is/ColazzoGS09, DBLP:journals/tcs/ColazzoGPS13, DBLP:journals/tods/ColazzoGPS13}, we restrict here to \emph{conflict-free} (CF) regular expressions, that are expressions where  i) any symbol may occur at most once (single-occurrence constraint), and ii) repetition */+ is only allowed over symbols. By using conflict-free expressions only, we can avoid the issues related to the nesting of repetitions as well as those concerning multiple occurrences of the same symbol.

Conflict-free expressions obey the following grammar:
\[
\begin{array}{lllll}
T  ::= &    \epsilon 
	  \ \mid\ a 
	  \ \mid\ T + T 
	  \ \mid\ T \ileave T
	  \ \mid\ \reps{a}
	  \ \mid\ \plus{a}
\end{array}
\]
and satisfy the  single-occurrence constraint: for any $T_{1} + T_{2}$ or $T_{1} \ileave T_{2}$ subexpression of a CF type, $\sym{T_{1}}\cap\sym{T_{2}}=\emptyset$ holds.

The expression $\reps{a} \ileave b + c $ is conflict-free, while the expression used in Example \ref{ex:nonempty} in schema element $\xeleme{1}$ is not, as the  single-occurrence constraint is not respected there; the expression $\reps{(a \ileave b)} \ileave c$ is another example of a non conflict-free expression: single-occurrence is met, but the restriction over repetitions is  not. 

Existing studies have shown that users tend to define CF expressions when creating schemas for XML data \cite{DBLP:conf/webdb/Choi02}. We believe that the same will hold in the context of data graphs as the reasons that lead users to adopt CF expressions  depend on aspects that are orthogonal to the the particular data model at hand: conflict-free expressions, indeed, have a semantics  that is relatively simple to understand by humans, and, at the same time, they allow one to describe and constrain a wide class of sequences that arise in the context of semi-structured data management. 
 
 The following example  shows that, unfortunately, conflict-freedom together with   properties that characterise schemas (Definition \ref{def:grsch}) are not sufficient to ensure non emptiness of graph schemas.

\begin{example}\label{ex:emptyCF}
Consider the simple  graph schema $\schemas = \{\xeleme{1}\}$, where:
\[
\begin{array}{llllllllll}
\xeleme{1} & = & (a + b, a \ileave b )
\end{array}
\]

In this schema each $\xeleme{1}$ node produces both a $b$ and an $a$ outgoing edge. The only nodes that can receive these edges are in turn of type  $\xeleme{1}$. These nodes, however,  can receive either a $b$ or an $a$ edge, and in turn emit other two $a$ and $b$ edges. This implies that no finite graph meets this schema. 
\end{example}

An alternative, and equivalent, formulation of the above schema is the following one, obtained by distributing element types over the union type in $a + b$ expression.
\[
\begin{array}{llllllllll}
\xeleme{1} & = & (a , a \ileave b )\\
\xeleme{2} & = & (b, a \ileave b )

\end{array}
\]

This formulation better highlights that, indeed, there are two kinds of nodes that can be generated by schema $\schemas$:  the first one is for nodes receiving an $a$-edge,  and the second one is for nodes receiving a $b$-edge. Now, since both $a$ and $b$ are emitted by both kinds, we could ensure non-emptiness by modifying the schema  as follows:
\[
\begin{array}{llllllllll}
\xeleme{1} & = & (\reps{a} , a \ileave b )\\
\xeleme{2} & = & (\reps{b}, a \ileave b )
\end{array}
\]

It is easy to verify that this schema is not empty and that infinitely many graphs conform to it. The idea underlying this modification is that, whenever a symbol $a$ is emitted by multiple schema elements in a schema, then each occurrence  of $a$ that appears in a receiving expression occurs under a $*$. This implies that there must exist a schema element whose vertices can accept as many $a$-edges as needed. 

As we will see, the generalisation and formalisation of the above sketched restriction actually ensures non-emptiness. Before switching to the formal treatment, it is worth stressing that this restriction demands that, whenever a symbol $b$ is emitted by multiple nodes $n_{1}, \ldots, n_{k}$ with different types, any  node $n$ receiving at least a $b$-edge is allowed by its type $e_{n}$ to have have multiple incoming $b$-edges, thus allowing $n$ to be  shared by  $n_1, \ldots, n_k$ via multiple $b$-edges.  

Note that, of course, a similar restriction is needed to for received symbols wrt emitted symbols: whenever a symbol $a$ is received  by multiple schema elements in a schema, then each occurrence  of $a$ that appears in a emitting expression occurs under a $*$. This rules out empty schemas like the one including the following node types (note that this schema is obtained from a previous one by simply swapping $in$ and $out$ expression). 
\[
\begin{array}{llllllllll}
\xeleme{1} & = & (a \ileave b  , a )\\
\xeleme{2} & = & ( a \ileave b , b )
\end{array}
\]

We identified this restriction after several other attempts with other restrictions. While proving non emptiness for these restrictions, the main problem we had was that a constructive approach (based on trying to build a graph valid wrt the schema in an incremental way)  failed because, each time a node of a given type was introduced, this could receive (emit) pending edges emitted (received) by other nodes already created in the process, but, at the same time, this new node introduced other constraints (pending outgoing and incoming edges) that existing nodes could  not satisfy. So this, in turn, triggered  the introduction of another node which  re-creates the same situation,  therefore leading to a circular and possibly non-terminating process. 

 We have identified the above depicted restriction in such a way that a terminating constructive approach can be used in the proof of non-emptiness.   While the restriction we adopt may seem artificial, we believe that it does not limit the modelling opportunities for the schema designer and that it can be safely adopted in automatic schema-learning approaches. Importantly, our restriction does not exclude schemas  describing graphs where some nodes can receive at most one edge with a given label, as illustrated by the following example.
 
 \begin{example} Consider graphs for representing social networks where users publish posts, which are commented and/or liked by other users, which in turn can establish friendship relationships with other users.\footnote{This is example is borrowed from the official neo4j web site (blog section).} A well formed schema for this database is  $\schemas = \{\xeleme{user}, \xeleme{post} \}$, where:
\[
\begin{array}{llllllllll}
\xeleme{user} & = & (\reps{friend},\  \reps{friend} \ileave \reps{post} \ileave \reps{commented}) \\[1mm]
\xeleme{post} & = & (post \ileave \reps{commented}  \ileave \reps{liked}, \epsilon)
\end{array}
\]
Of course, a user can publish several posts, while a post is posted by only one user. 
 \end{example}

In order to formalize the above illustrated restriction, we have first to normalize regular expressions. Regular expressions must be transformed in \emph{Disjunctive Normal Form} (\sdnf), and then the whole schema must be normalized (as illustrated before) in order to distribute unions over type definitions. The following example illustrates why normalisation is necessary for proving non-emptiness. 

\begin{example}\label{ex:emptyCF2}
Consider the simple  graph schema $\schemas = \{\xeleme{1}\}$, where:
\[
\begin{array}{llllllllll}
\xeleme{1} & = & (a\ileave b \ileave c, a \ileave ( b + c) )
\end{array}
\]

It can be easily proved (by contradiction) that this schema is empty.  However, in its current formulation this schema satisfies the restriction sketched above, as each symbol occurs once in every incoming/outgoing regular expression and no repetition is used, hence contradicting our previous claim. This is due to the fact that, as in Example \ref{ex:emptyCF}, the current formulation hides that the schema actually defines two kinds of nodes. In order to exhibit this property, the outgoing regular expression must be normalised, thus obtaining the following schema:
\[
\begin{array}{llllllllll}
\xeleme{1} & = & (a\ileave b \ileave c, (a \ileave b) + (a \ileave c) )
\end{array}
\]

Furthermore,  the whole schema must be transformed in order to distribute element type definitions over the union emerged by means of normalisation, thus obtaining:

\[
\begin{array}{llllllllll}
\xeleme{1} & = & (a\ileave b \ileave c, a \ileave b  ) \\
\xeleme{2} & = & (a\ileave b \ileave c,  a \ileave c )
\end{array}
\]
As it can be observed, this schema formulation  does not satisfy our restriction. 
\end{example}

\begin{definition}[Disjunctive normal form]\label{def:dnf}
A regular expression $T$ is in Disjunctive Normal Form ({\sdnf}) if it obeys the following grammar.
\begin{eqnarray*}
T & ::= & C + \ldots + C\\
C & ::= & \epsilon \ \mid \ a \ \mid \ C \ileave C \ \mid \ a^{*} \mid \ \plus{a}
\end{eqnarray*}
\end{definition}

Any regular expression can be transformed in {\sdnf} by using the function defined below, where $C_{i}$ denotes a union-free regular expression, and $\bigcup_{i=1}^{n} T_i$ denotes  $T_1 \ileave \ldots \ileave T_n$.

\begin{definition}[$\xnorm{\cdot}$]\label{def:norm}
\[
\begin{array}{llllllllll}
(1) & \xnorm{\epsilon} & = & \epsilon \\
(2) & \xnorm{a} & = & a \\
(3) & \xnorm{T_{1} + T_{2}} & = & \xnorm{T_{1}} + \xnorm{T_{2}} \\
(4) & \xnorm{T_{1} \ileave T_{2}} & = & \xnorm{\bigcup_{i = 1}^{n} \bigcup_{j = 1}^{m} A_{i} \ileave B_{j}} & \\
      &    &&  \ \ \ \text{where $\xnorm{T_{1}} = \bigcup_{i = 1}^{n} A_{i}$}\\
& & &  \ \ \ \text{and  $\xnorm{T_{2}} = \bigcup_{j = 1}^{m} B_{j}$}\\
(5) & \xnorm{\reps{a}} & = & \reps{a} \\
(6) & \xnorm{\plus{a}} & =  & \plus{a}
\end{array}
\]
\end{definition}
It is easy to prove that $T$ and $\xnorm{T}$ are equivalent for any $T$. To prove that $\xnorm{T}$ actually transforms any regular expression in  disjunctive normal form, we need a preliminary lemma.

\renewcommand{\sdnf}{DNF}
\DC{le prove dei due lemmi che seguono devono essere semplificate dal momento che normalizziamo tipi CF}
\CS{Ho iniziato a semplificarle.}
\begin{lemma}\label{lem:normprel}
Given a CF regular expression $T$, if $T$ contains a single union, then $\xnorm{T}$ is in \sdnf.
\end{lemma}
\begin{proof}
We prove the thesis by induction on the level of the parse tree of $T$ where the union is located. Assume that the parse tree of $T$ contains $n+1$ levels, where  0 is the level of the root.
\begin{description}
\item[Base] Assume that the union is located on the root of the parse tree (level 0). Then $T = T_{1} + T_{2}$, where $T_{1}$ and $T_{2}$ are union-free. Hence, $T$ is already in \sdnf\ and $\xnorm{T} = \xnorm{T_{1}} + \xnorm{T_{2}} = T_1 + T_ 2$ is in \sdnf.

\item[Inductive step] Assume that the thesis is true for any regular expression containing a single union $T_{1} + T_{2}$ at level $i$ and assume  that $T$ contains a single union at level $i + 1$. Then, if we indicate with $T^{\prime}$ the subterm at level $i$ surrounding $T_{1} + T_{2}$, $T^{\prime} = (T_{1} + T_{2}  \ileave T_{3})$, where $T_{1}$, $T_{2}$, and $T_{3}$ are union-free. In this case, by applying rule (4) of Definition \ref{def:norm}, $\xnorm{T^{\prime}} = \xnorm{T^{\prime \prime}}$, where $T^{\prime \prime} = (\xnorm{T_{1}} \ileave \xnorm{T_{3}}) + (\xnorm{T_{2}} \ileave \xnorm{T_{3}})$. In $T^{\prime \prime}$ the union is at level $i$, hence, by induction, $\xnorm{T^{\prime \prime}}$ is in \sdnf, which proves the thesis.
\end{description}
\end{proof}

\begin{lemma}\label{lem:norm}
Given a CF regular expression $T$, $\xnorm{T}$ is in \sdnf.
\end{lemma}
\begin{proof}
We prove the thesis by induction on the number of unions inside $T$.
\begin{description}
\item[Base] If $T$ contains a single union, the thesis is true by Lemma \ref{lem:normprel}.

\item[Inductive step] We assume that the thesis is true for regular expressions containing $n$ unions. Let $T$ be a regular expression containing $n+1$ unions. We proceed by induction on the level of the topmost union.

If the topmost union is at level 0, then $T = T_{1} + T_{2}$, where $T_{1}$ and $T_{2}$ contain at most $n$ unions; by induction, $\xnorm{T_{1}}$ and $\xnorm{T_{2}}$ are in \sdnf. Therefore, $\xnorm{T} = \xnorm{T_{1}} + \xnorm{T_{2}}$ is in \sdnf.

Assume now that the topmost union $T_{1} + T_{2}$  is at level $i + 1$ and the thesis is true for level $i$. Then, if $T^{\prime}$ is the subterm at level $i$ containing $T_{1} + T_{2}$, $T^{\prime} = (T_{1} + T_{2}) \ileave T_{3}$. $T_{1}$, $T_{2}$, and $T_{3}$ contain at most $n$ unions; hence, by the outer induction, $\xnorm{T_{1}}$, $\xnorm{T_{2}}$, and $\xnorm{T_{3}}$ are in \sdnf. If $T^{\prime} = (T_{1} + T_{2}) \ileave T_{3}$, then, by applying rule (4) of Definition \ref{def:norm}, $\xnorm{T^{\prime}} = \xnorm{T^{\prime \prime}}$, where $T^{\prime \prime} = (\xnorm{T_{1}} \ileave \xnorm{T_{3}}) + (\xnorm{T_{2}} \ileave \xnorm{T_{3}})$; $\xnorm{T^{\prime}}$, hence, lifts the union to level $i$; by inner induction, we have the thesis.

\end{description}
\end{proof}


\newcommand{\xxnorm}[1]{\textit{DNorm}(#1)}
\newcommand{\enorm}[1]{\textit{enorm}(#1)}

\begin{definition}[$\xxnorm{\schemas}$]\label{def:dnorm}
Given a schema $\schemas$, we indicate with $\xxnorm{\schemas}$ its \emph{double normalisation}, i.e.,  the schema obtained from $\schemas$ by first normalising each regular expression in  $\xeleme{i}$'s in $\schemas$, and then by distributing element types over unions of normalised expressions in  $\xeleme{i}$'s. 

Formally,    $\xxnorm{\schemas}$ is the set of element types  $\xeleme{{{i,j}}}^k = (C_i , D_j)$ such that:  there exists  $\xeleme{k}\ = \ (in , out)$ in  $\schemas$ with $\xnorm{in}=C_1 +\ldots +C_n$ and $\xnorm{out}=D_1 +\ldots +D_m$, and  $i\in[1\ldots n]$ and $j\in[1\ldots m]$. 
\end{definition}

We can now introduce the class of \emph{well-formed schemas} corresponding to our restriction.  
\begin{definition}[Well-formed schemas] A schema $\schemas$ is well-formed if  the following holds.

For any symbol $a$, if there exist $\xeleme{1}=(in_1, out_1)$ and  $\xeleme{2}=(in_2, out_2)$ in $\xxnorm{\schemas}$ such that $a$ occurs in both $out_1$ and $out_2$ ($in_1$ and $in_2$) then, for any   $\xeleme{}=(in, out)$ in $\xxnorm{\schemas}$, any occurrence of $a$ in   $in$ ($out$, respectively) must be under a $*$.

\end{definition}

\begin{theorem} Every well-formed schema $S$ is not empty.

\end{theorem}

\begin{proof} We first observe that $\schemas$ and $\xxnorm{\schemas}$ are equivalent,  that (*) each type in $\xxnorm{\schemas}$ contains only union-free and CF $in/out$ expressions, and  that (**) $\xxnorm{\schemas}$ respects the  properties 1 and 2  of Definition \ref{def:grsch} (these properties can be easily proved). We then prove that $\xxnorm{\schemas}$ is not empty. To this end we prove that any schema $\schemas^{\prime}$ satisfying (*) and (**) has a valid graph $G$ having exactly one node for  each element type in $\schemas^{\prime}$ and  that, for each node of $G$ having type $\eleme = (in, out)$ in $\schemas^{\prime}$, there exists an incoming/outgoing $a$-edge for each $a$ in $in/out$. Clearly this property entails the desired one.  

We proceed by induction on $|\schemas^{\prime}|$, i.e., the number of schema elements in $\schemas^{\prime}$.

For the  base case $|\schemas^{\prime}|=1$ we can build a graph with only one node with an $a$ incoming (outgoing) edge for \emph{each} symbol $a$ in the regular  expressions  of the only type   $\xeleme{1}=(in, out)$ of $\xxnorm{\schemas}$. The fact that this  graph is valid with respect to $\xxnorm{\schemas}$ follows from (*) and (**).

Let us consider now the case that $|S'|=n+1$ with $n>1$.
We pick a type $\xeleme{}=(in, out)$ in $S'$ and build  a schema $\schemas'_{-\xeleme{}}$ by dropping out $\xeleme{}=(in, out)$ from $\schemas^{\prime}$ and by deleting, in expressions of remaining types, every symbol  that  occurs only in  $\xeleme{}=(in, out)$. The schema $S'_{-\xeleme{}}$  still satisfies properties (*) and (**), so by induction we can assume that there exists a graph $G$ conforming to it  and that (***) each type in  $\schemas'_{-\xeleme{}}$ has exactly one corresponding node in $G$ and that each node of $G$ having type $\eleme'= (in', out')$ there exists an incoming/outgoing $a$-edge for each $a$ in the $in'/out'$.

Now, we add a new node $n$ node to $G$ as follows. The new node $n$ will contain an incoming (outgoing) $a$-edge \emph{for each} symbol $a$ in the regular regular expression $in$ ($out$)  of the dropped type   $\xeleme{}=(in, out)$. In addition,  we reactivate erased symbols in types of $\schemas'_{-\xeleme{}}$ and add a corresponding  incoming/outgoing edge in each node of a type $\xeleme{}'=(in', out')$ in $\schemas'_{-\xeleme{}}$ for which the symbol has been reactivated in $in' / out'$. 

\DC{forse ci vuole un esempio per illustrare re-activate}

At this point, it may happen that some of the added edges are dangling. We will show that  we can connect these edges to exiting nodes (including $n$)  thanks to the following facts. Properties   1 and 2  of Definition \ref{def:grsch}, plus (***)  ensure that   for each new dangling  $a$-edge either i) there is  an existing (not dangling)  $a$-edge $g$ connecting two nodes $n_{1}$ and $n_{2}$ of $G$ (this is the case when the added edge is outgoing/incoming for $n$ and the edge symbol is already used in types in $\schemas'_{-\xeleme{}}$), or ii) the  dangling edge has a  reactivated label $a$ (a label used in $e$ but not in $\schemas'_{-\xeleme{}}$) and actually multiple of such pending $a$-edges can exist. 


For the first case i) we can distinguish two sub-cases. 

The first one  is that the dangling $a$-edge is outgoing from $n$. Recall that  types corresponding to $n$, $n_1$ and $n_2$ are different types $\xeleme{}, \xeleme{1}, \xeleme{2}$. Wlog, assume that the existing  $a$-edge is from $n_1$ to $n_2$.  Before proceeding, observe that at this point we have: $\xeleme{1}=(in_1, out_1)$ and $a$ occurs in $out_1$, $\xeleme{}=(in, out)$  and $a$ occurs in $out$, and  $\xeleme{2}=(in_2, out_2)$ and $a$ occurs in $e_{2}$. Since $S$ is well-formed, this means that  $a$ occurs in $in_2$ under a $*$, and this implies that the dangling edge can be connected to $n_{2}$.

The second case is that the dangling edge is an incoming  edge of  $n$. This case can be proved as above.

\medskip
Concerning the case ii) we can distinguish the following two sub-cases. 

The first one deals with one or more pending  $a$-edges that are  outgoing from a node of $G$ or from  $n$. Recall that  $a$  is in $\xeleme{}$  but in no other type of $\schemas'_{-\xeleme{}}$. 

Now, we observe that  these edges originates from nodes of different types of $\schemas'$, thanks to (***) and to the fact that  $n$ has exactly one edge for each symbol of $e$; recall that $a$ has been reactivated. Thanks to  (**)  we have that there must be a type $\eleme'=(in', \ out') $ in $\schemas'$ with $a\in \sym{in'}$. In the case that only one pending $a$-edge exists, the case is proved  since the edge can be connected to the node having type $\eleme'$ (observe that it may be the case that a pending \emph{incoming} $a$-edge has been added for this node; in this case the two edges are simply merged). In the remaining case is we have more than one pending $a$-edge; in this case thanks to (**), (***) and to well-formedness of $\schemas$, we have that  there must be a type $\eleme'=(in', \ out') $ in $\schemas'$ with $a\in \sym{in'}$ such that $a$ is under a $*$. So the node having type $e'$ is able to receive all of the pending outgoing $a$-edges. 

The second ii) sub-case deals with one or more pending incoming  $a$-edges  targeting either a node of $G$ or  $n$.  This case is similar to the previous one. \DC{force questo caso lo devo spiegare}
\end{proof}


\section{RPQs, NREs and GXPath}\label{sec:querylan}

RPQs, NREs, and GXPath are graph query languages based on the idea of using regular expressions to specify patterns that must be matched by paths in the input graph. Given a query $q$, the result of its evaluation over a graph $G$ is always a set of node pairs $(v, v^{\prime})$ such that $v$ and $v^{\prime}$ are connected by a path $p$ in $G$ matching the query $q$.

These languages mainly differ in the class of supported regular expressions, ranging from standard regular expressions to expressions with counters and nested predicates.

\DC{che facciamo diamo un po' di referenze e una discussione per sottolineare l'importanza di RPQ nella letteratura e in pratica ? Potrebbe rinforzare il peso dei nostri risultati.}

\medskip

\emph{Regular Path Queries} (RPQs) are the most basic language we are analyzing here. Given a finite alphabet $\Sigma$, an RPQ $r$ over $\Sigma$ is defined by the following grammar:
\[
\begin{array}{lllllll}
r & ::= & \epsilon \mid a \mid r + r \mid r \cdot r \mid r^{*}
\end{array}
\]

Given a graph $G = (V, E, \rho)$, the semantics of RPQs can be defined as follows.
\begin{eqnarray*}
\ggsem{\epsilon} & = & \{ (u,u) \mid u \in V \} \\
\ggsem{a} & = & \{ (u,v) \mid (u,a,v) \in E \} \\
\ggsem{r_{1} + r_{2}} & = & \ggsem{r_{1}} \cup \ggsem{r_{2}} \\
\ggsem{r_{1} \cdot r_{2}} & = & \ggsem{r_{1}} \rconc \ggsem{r_{2}} \\
\ggsem{r^{*}} & = & \cup_{i \geqslant 0} \ggsem{r}^{i}
\end{eqnarray*}

where $\rconc$ is the symbol for the concatenation of binary relations and $R^{i}$ denotes the concatenation of $R$ with itself $i$ times.

\DC{$q$ o $r$ per le query? Ho usato $q$ in diverse prove...}

\begin{example}
Consider the schema of Example \ref{ex:second} and the following query:
$$partOf \cdot series$$

This query selects all the conference papers and relates them to the corresponding conference series. In the case of the graph of Example \ref{ex:first}, the result is the following:\footnote{We use here the value of a node to denote the node itself.}
$$\{(HopcroftU67a, focs) \}$$
\end{example}

As it can be seen from the example, RPQs  can express neither branching nor backward navigation, which are introduced by NREs queries.

\medskip

\emph{Nested Regular Expressions} (NREs) are an evolution of RPQs and form the basis of the path language of SparQL \cite{Sparql}. NREs introduce the ability of traversing edges backwards, as in 2RPQs \cite{DBLP:conf/pods/CalvaneseGLV00}, as well as the ability of specifying conditions inside paths.

NREs obey the following grammar:
\[
\begin{array}{lllllll}
n & ::= & \epsilon \mid a \mid \back{a} \mid n + n \mid n \cdot n \mid n^{*} \mid [n]
\end{array}
\]where $\back{a}$ denotes a backward navigation and $[n]$ allows one to express conditions inside a path expression. Given a graph $G = (V, E, \rho)$, the semantics of NREs can be defined as follows.
\begin{eqnarray*}
\ggsem{\epsilon} & = & \{ (u,u) \mid u \in V \} \\
\ggsem{a} & = & \{ (u,v) \mid (u,a,v) \in E \} \\
\ggsem{\back{a}} & = & \{ (u,v) \mid (v,a,u) \in E \} \\
\ggsem{n_{1} + n_{2}} & = & \ggsem{n_{1}} \cup \ggsem{n_{2}} \\
\ggsem{n_{1} \cdot n_{2}} & = & \ggsem{n_{1}} \rconc \ggsem{n_{2}} \\
\ggsem{n^{*}} & = & \cup_{i \geqslant 0} \ggsem{n}^{i}\\
\ggsem{[n]} & = & \{ (u,u) \mid (u,v) \in \ggsem{n} \}
\end{eqnarray*}

\begin{example}\label{ex:nre}
Consider again the graph of Example \ref{ex:first} and the schema of Example \ref{ex:second}. The following query returns all pairs $(x,y)$ where $x$ is the author of a paper in a conference series $y$, but also published a paper in a journal $z$:
$$[\back{creator} \cdot journal] \cdot \back{creator} \cdot partOf \cdot series$$

The result of this query is $\{(John E. Hopcroft, focs)\}$. Observe that this query cannot be expressed through RPQs or 2RPQs.
\end{example}

GXPath is the most powerful language we are examining here and has been recently proposed by Libkin et al. in \cite{DBLP:conf/icdt/LibkinMV13}. GXPath is essentially an adaptation of XPath to data graphs. Wrt the previous languages, GXPath introduces the \emph{complement} operator, data tests on the values stored into nodes, as well as counters, which generalize the Kleene star.

Among the various fragments of GXPath, we focus here on the navigational, path-positive fragment with intersection, described by the following grammar.
\[
\begin{array}{lllllll}
\alpha & ::= & \epsilon \mid \_ \mid a \mid \back{a} \mid \alpha + \alpha \mid \alpha \cdot \alpha  \mid \xcount{\alpha}{m}{n} \mid \alpha \cap \alpha \mid [\alpha]
\end{array}
\]

Given a graph $G = (V, E, \rho)$, the semantics of GXPath can be defined as follows.
\begin{eqnarray*}
\ggsem{\epsilon} & = & \{ (u,u) \mid u \in V \} \\
\ggsem{\_} & = & \{(u,v) \mid \exists a \in \Sigma . (u,a,v) \in E \}\\
\ggsem{a} & = & \{ (u,v) \mid (u,a,v) \in E \} \\
\ggsem{\back{a}} & = & \{ (u,v) \mid (v,a,u) \in E \} \\
\ggsem{\alpha_{1} + \alpha_{2}} & = & \ggsem{\alpha_{1}} \cup \ggsem{\alpha_{2}} \\
\ggsem{\alpha_{1} \cdot \alpha_{2}} & = & \ggsem{\alpha_{1}} \rconc \ggsem{\alpha_{2}} \\
\ggsem{\xcount{\alpha}{m}{n}} & = & \cup_{i = m}^{n} \ggsem{\alpha}^{i}\\
\ggsem{\alpha_{1} \cap \alpha_{2}} & = & \ggsem{\alpha_{1}} \cap \ggsem{\alpha_{2}}\\
\ggsem{[\alpha]} & = & \{ (u,u) \mid (u,v) \in \ggsem{\alpha} \}
\end{eqnarray*}

\begin{example}\label{ex:qcycle}
Consider the graph depicted in Figure \ref{fig:secgraph}.

\begin{figure}
\centering
\begin{tikzpicture}[->,>=stealth',shorten >=1pt,auto,node distance=1cm,
                    thick,main node/.style={rounded rectangle,draw,font=\sffamily\tiny}]

  \node[main node] (1) at (1,0) {1};
  \node[main node] (2) at (0,-1) {3};
  \node[main node] (3) at (2,-1) {7};
  \node[main node] (4) at (1,-2) {1};
  \node[main node] (5) at (3,-2) {5};
  \node[main node] (6) at (0,-3) {2};
  \node[main node] (7) at (2,-3) {3};

  \path[every node/.style={font=\sffamily\small}]
    (1) edge[swap] node {a} (2)
    (2) edge[swap] node {a} (4)
    (3) edge[swap] node {a} (1)
         edge node {d} (5)
    (4) edge[swap] node {a} (3)
         edge[swap] node  {b} (6)
         edge node  {c} (7);
\end{tikzpicture}
\caption{A graph.}
\label{fig:secgraph}
\end{figure}

Consider now the following query: $[\_ \cdot (\reps{\_}) \cap \epsilon] \cdot (b + c)$. This query selects all pairs of nodes $(x,y)$ where $x$ is part of a cycle and $y$ is reachable from $x$ through  an edge labelled with $b$ or $c$. 
\end{example}

In the following we will indicate with \rpq , \nre , and \gxp\ the three classes of regular expressions we are studying here.



\section{Inference Rules}\label{sec:rules}

In this section we present a  type inference approach for typing RPQs, NREs, and GXPath queries. 
The approach we propose here is a basic  yet useful one. It associates to each query a set of schema element pairs; hence, a query $q$ is typed by a set $\xset{(\elemei, \elemei^{\prime})}{i}$, where $\elemei$ and $\elemei^{\prime}$ are schema elements describing the nodes at the beginning and at the end of path $p$ matching $q$.  
Another advantage of this typing approach is that can be performed in polynomial time (Theorem \ref{theo:bcomplx}) and that it is sound and complete for RPQs (Theorems \ref{theo:basicsound} and \ref{theo:rpqcompl}). For  NRE and GXPath queries only soundness holds (we will provide counterexamples for completeness).

Typing rules rely on the judgement defined below. We use the meta-variables $\psete$ and $\psetei$ to denote sets of schema element pairs.

\begin{definition}[Basic Judgment]
$\judbasic{\schemas}{q}{\psete}$ is  a judgment stating that, given a well-formed  $\schemas$ and a  graph $G \in \ssem{\schemas}$, $\psete$ is an upper bound for $\ggsem{q}$.
\end{definition}

Type inference rules for RPQs, NREs, and GXPath queries are shown in Tables \refDisplayHere{table:basicrpqrules}, \refDisplayHere{table:basicnrerules}, and \refDisplayHere{table:basicgxprules}. In these rules, $\rconc$ is the operator for the usual combination of binary relations, and $\first{\psete} = \{ \elemei \mid \exists \elemej . (\elemei, \elemej) \in \psete \}$. In Table \refDisplayHere{table:basicrpqrules} rule (\textsc{TypeEpsilon}) types $\epsilon$ queries, rule (\textsc{TypeLabel}) deals with forward navigation, while rules (\textsc{TypeUnion}) and (\textsc{TypeConc}) type queries with union and concatenation, respectively. Rule (\textsc{TypeStar}), finally, deals with $r^{*}$ queries.

In Table \refDisplayHere{table:basicnrerules}, rules (\textsc{TypeBackLabel}) and (\textsc{TypeCond}) infer a type for  queries with backward navigation and nested conditions, respectively.

In Table \refDisplayHere{table:basicgxprules}, finally, rules (\textsc{TypeAnyLabel}), (\textsc{TypeCount}), and (\textsc{TypeIntersect}) deal with, respectively, wildcard queries, counting, and intersection.

{\footnotesize
\begin{figure}[htb]
{\labelDisplay{table:basicrpqrules}
\begin{display}{Basic  inference rules for RPQs.}
a\=aaaaaaaaaaaaaaaaaaaaaaaaaaaaaaa\iftimes{aaaaaa}
aa \= aaaaaaaaaaaaaaaaaaaaaaaaaaaaaaaaaaaaaaaaaaaa\=aaaaaaaaaaaaaaaaaaaaaaaaaaaaaaaaaaaa
\kill \> \staterule{(\textsc{TypeEpsilon})}
  {\elemei \in \schemas}
  {\judbasic{\schemas}{\epsilon}{\xset{(\elemei, \elemei)}{i}}}
\\[0.2cm]
\> \staterule{(\textsc{TypeLabel})}
  {\elemei \in \schemas \ \ \elemei^{\prime} \in \schemas \ \ \ls{a}{k} \in \sym{\gout{\elemei}} \cap \sym{\gin{\elemei^{\prime}}}}
  {\judbasic{\schemas}{a}{\xset{(\elemei, \elemei^{\prime})}{i}}}
\\[0.2cm]
\>  \staterule{(\textsc{TypeUnion})}
  {\judbasic{\schemas}{r_{1}}{\psete_{1}} \ \ \judbasic{\schemas}{r_{2}}{\psete_{2}}}
  {\judbasic{\schemas}{r_{1} + r_{2}}{\psete_{1} \cup \psete_{2}}}

\> \staterule{(\textsc{TypeConc})}
  {\judbasic{\schemas}{r_{1}}{\psete_{1}} \ \ \judbasic{\schemas}{r_{2}}{\psete_{2}}}
  {\judbasic{\schemas}{r_{1} \cdot r_{2}}{\psete_{1} \rconc \psete_{2}}}

\\[0.2cm]
\> \staterule{(\textsc{TypeStar})}
  {\judbasic{\schemas}{r}{\psete}}
  {\judbasic{\schemas}{r^{*}}{\bigcup_{i \geqslant 0} \psete^{i}}}
\end{display}}
\end{figure}
}

{\footnotesize
\begin{figure}[htb]
{\labelDisplay{table:basicnrerules}
\begin{display}{Additional basic  inference rules for NREs.}
a\=aaaaaaaaaaaaaaaaaaaaaaaaaaaaaaaaaaa\iftimes{aaaaaa}
aa \= aaaaaaaaaaaaaaaaaaaaaaaaaaaaaaaaaaaaaaaaaaaa\=aaaaaaaaaaaaaaaaaaaaaaaaaaaaaaaaaaaa
\kill \> \staterule{(\textsc{TypeBackLabel})}
  {\elemei \in \schemas \ \ \elemei^{\prime} \in \schemas \ \ \ls{a}{k} \in \sym{\gin{\elemei}} \cap \sym{\gout{\elemei^{\prime}}}}
  {\judbasic{\schemas}{\back{a}}{\xset{(\elemei, \elemei^{\prime})}{i}}}
\\[0.2cm]
\>  \staterule{(\textsc{TypeCond})}
  {\judbasic{\schemas}{n}{\psete_{1}}}
  {\judbasic{\schemas}{[n]}{\first{\psete_{1}} \times \first{\psete_{1}}}}
\end{display}}
\end{figure}
}

{\footnotesize
\begin{figure}[htb]
{\labelDisplay{table:basicgxprules}
\begin{display}{Additional basic  inference rules for GXPath queries.}
a\=aaaaaaaaaaaaaaaaaaaaaaaaaaaaaaaaaaa\iftimes{aaaaaa}
aa \= aaaaaaaaaaaaaaaaaaaaaaaaaaaaaaaaaaaaaaaaaaaa\=aaaaaaaaaaaaaaaaaaaaaaaaaaaaaaaaaaaa
\kill \> \staterule{(\textsc{TypeAnyLabel})}
  {\elemei \in \schemas \ \ \elemei^{\prime} \in \schemas \ \ \sym{\gout{\elemei}} \cap \sym{\gin{\elemei^{\prime}}} \neq \emptyset}
  {\judbasic{\schemas}{\_}{\xset{(\elemei, \elemei^{\prime})}{i}}}
\\[0.2cm]

\> \staterule{(\textsc{TypeCount})}
  {\judbasic{\schemas}{\alpha}{\psete}}
  {\judbasic{\schemas}{\xcount{\alpha}{m}{n}}{\bigcup_{i = m}^{n} \psete^{i}}}

\\[0.2cm]
\> \staterule{(\textsc{TypeIntersect})}
  {\judbasic{\schemas}{\alpha_{1}}{\psete_{1}} \ \ \judbasic{\schemas}{\alpha_{2}}{\psete_{2}}}
  {\judbasic{\schemas}{\alpha_{1} \cap \alpha_{2}}{\psete_{1} \cap \psete_{2}}}
\end{display}}
\end{figure}
}

\begin{example}\label{ex:basicinf}
Consider again the query of Example \ref{ex:nre}. To type this query, rules (\textsc{TypeConc}) and (\textsc{TypeCond}) are first invoked; rule (\textsc{TypeCond}), in turn, invokes rules (\textsc{TypeConc}), (\textsc{TypeLabel}), and (\textsc{TypeBackLabel}) to type $\back{creator} \cdot journal$.  Rule (\textsc{TypeBackLabel}) returns the set $\{(\xeleme{5}, \xeleme{1})\}$, while rule (\textsc{TypeLabel}) returns $\{ (\xeleme{1}, \xeleme{2}) \}$. Rule (\textsc{TypeConc}), hence, returns $\{(\xeleme{5},\xeleme{2})\}$, while rule (\textsc{TypeCond}) returns $\{(\xeleme{5},\xeleme{5})\}$.

Rules (\textsc{TypeConc}), (\textsc{TypeLabel}), and (\textsc{TypeBackLabel}) are called again to type $\back{creator} \cdot partOf \cdot series$, returning $\{(\xeleme{5},\xeleme{1})\} \rconc \{(\xeleme{1},\xeleme{3}) \} \rconc \{(\xeleme{3},\xeleme{4}) \} = \{(\xeleme{5}, \xeleme{4})\}$.

Therefore, the result of the type inference is $\{(\xeleme{5}, \xeleme{5})\} \rconc \{(\xeleme{5}, \xeleme{4})\} = \{(\xeleme{5}, \xeleme{4})\}$, as expected.
\end{example}

The soundness of basic type inference is stated by the following theorem.

\begin{theorem}\label{theo:basicsound}
Given a well-formed  $\schemas$ and a query $q$ on a graph $G = (V, E , \rho) \in \ssem{\schemas}$, if $\judbasic{\schemas}{q}{\psete}$, then, for each $(u,v) \in \ggsem{q}$ there exists $(\xeleme{i}, \xeleme{j}) \in \psete$ such that $u \in \ssem{\xeleme{i}}$ and $v \in \ssem{\xeleme{j}}$.
\end{theorem}

\begin{proof}
By structural induction on the queries.

\begin{description}
\item[($q = \epsilon$)]  Trivial.

\item[($q = a$)]  Let $(u,v) \in \ggsem{a}$. By definition of query semantics, there exists $(u,a,v) \in E$. By definition of well-formed schemas, there exists $\xeleme{i}, \xeleme{j} \in \schemas$ such that $u \in \ssem{\xeleme{i}}$, $v \in \xeleme{j}$, and $\ls{a}{k} \in \sym{\gout{\xeleme{i}}} \cap \sym{\gin{\xeleme{j}}}$ for some $k$. By rule (\textsc{TypeLabel)}, $(\xeleme{i}, \xeleme{j}) \in \psete$.

\item[($q = \back{a}$)]   Let $(u,v) \in \ggsem{\back{a}}$. By definition of query semantics, there exists $(v,a,u) \in E$. By definition of well-formed schemas, there exists $\xeleme{i}, \xeleme{j} \in \schemas$ such that $u \in \ssem{\xeleme{i}}$, $v \in \xeleme{j}$, and $\ls{a}{k} \in \sym{\gin{\xeleme{i}}} \cap \sym{\gout{\xeleme{j}}}$ for some $k$. By rule (\textsc{TypeBackLabel)}, $(\xeleme{i}, \xeleme{j}) \in \psete$.

\item[($q = \_$)]  Let $(u,v) \in \ggsem{\_}$. By definition of query semantics, there exists $a \in V$ such that  $(u,a,v) \in E$. By definition of well-formed schemas, there exists $\xeleme{i}, \xeleme{j} \in \schemas$ such that $u \in \ssem{\xeleme{i}}$, $v \in \xeleme{j}$, and $\ls{a}{k} \in \sym{\gout{\xeleme{i}}} \cap \sym{\gin{\xeleme{j}}}$ for some $k$. By rule (\textsc{TypeAnyLabel)}, $(\xeleme{i}, \xeleme{j}) \in \psete$.

\item[($q = r_{1} + r_{2}$)]  Let $(u,v) \in \ggsem{r_{1} + r_{2}}$. By definition of query semantics, $(u,v) \in \ggsem{r_{1}}$ or $(u,v) \in \ggsem{r_{2}}$. Wlog we can assume that $(u,v) \in \ggsem{r_{1}}$ (the case for $(u,v) \in \ggsem{r_{2}}$ is symmetrical). By induction, there exist $\xeleme{i}, \xeleme{j} \in \schemas$ such that $u \in \ssem{\xeleme{i}}$, $v \in \ssem{\xeleme{j}}$, and $(u,v) \in \psete_{1}$, where $\judbasic{\schemas}{r_{1}}{\psete_{1}}$. By rule (\textsc{TypeUnion}), $(\xeleme{i}, \xeleme{j}) \in \psete$.

\item[($q = r_{1} \cdot r_{2}$)] Let $(u,v) \in \ggsem{r_{1} \cdot r_{2}}$. By definition of query semantics, there exists $w \in V$ such that $(u,w) \in \ggsem{r_{1}}$ and $(w,v) \in \ggsem{r_{2}}$. Let $\judbasic{\schemas}{r_{1}}{\psete_{1}}$ and $\judbasic{\schemas}{r_{2}}{\psete_{2}}$. By induction and by uniqueness of vertex typing, there exist $\xeleme{i}, \xeleme{j}, \xeleme{k} \in \schemas$ such that $u \in \ssem{\xeleme{i}}$, $w \in \ssem{\xeleme{j}}$, $v \in \ssem{\xeleme{k}}$, $(\xeleme{i}, \xeleme{j}) \in \psete_{1}$, and $(\xeleme{j}, \xeleme{k}) \in \psete_{2}$. By rule (\textsc{TypeConc}), $(\xeleme{i}, \xeleme{j}) \rconc (\xeleme{j}, \xeleme{k}) = (\xeleme{i}, \xeleme{k}) \in \psete$.

\item[($q = r^{*}$)]  The thesis follows from the soundness of typing of union and concatenation.

\item[($q = [n\xbra$)]  Let $(u,u) \in \ggsem{[n]}$. By definition of query semantics, there exists $v \in V$ such that $(u,v) \in \ggsem{n}$. Let $\judbasic{\schemas}{n}{\psete_{1}}$. By induction, there exist $\xeleme{i}, \xeleme{j} \in \schemas$ such that $u \in \ssem{\xeleme{i}}$, $v \in \ssem{\xeleme{j}}$, and $(\xeleme{i}, \xeleme{j}) \in \psete_{1}$. By rule (\textsc{TypeCond}), $(\xeleme{i}, \xeleme{j}) \in \psete$.

\item[($q = \xcount{\alpha}{m}{n}$)]  The thesis follows from the soundness of typing for union and concatenation.

\item[($q = \alpha_{1} \cap \alpha_{2}$)]  Let $(u,v) \in \ggsem{\alpha_{1} \cap \alpha_{2}}$. By definition of query semantics, $(u,v) \in \ggsem{\alpha_{1}}$ and $(u,v) \in \ggsem{\alpha_{2}}$. Let $\judbasic{\schemas}{\alpha_{1}}{\psete_{1}}$ and $\judbasic{\schemas}{\alpha_{2}}{\psete_{2}}$. By induction and by uniqueness of node typing, there exist $\xeleme{i}, \xeleme{j} \in \schemas$ such that $u \in \ssem{\xeleme{i}}$, $v \in \ssem{\xeleme{j}}$, $(\xeleme{i}, \xeleme{j}) \in \psete_{1}$, and $(\xeleme{i}, \xeleme{j}) \in \psete_{2}$. By rule (\textsc{TypeIntersect}), $(\xeleme{i}, \xeleme{j}) \in \psete$. 
\end{description}

\end{proof}

The basic type inference approach returns quite simple information. This fact is counterbalanced by its polynomial complexity, as stated by the following theorem.

\begin{theorem}[Complexity of $\judbasic{\schemas}{q}{\psete}$]\label{theo:bcomplx}
$\judbasic{\schemas}{q}{\psete}$ can be evaluated in polynomial time.
\end{theorem}

\begin{proof}[sketch]
To prove the thesis we must first observe that, given a query $q$ of length $|q|$, each rule consumes at least one node in the parsing tree of $q$. This implies that $q$ will be typed by a number of rule invocations polynomial in $|q|$.

To complete the proof, it suffices to prove that each rule can be evaluated in polynomial time. This proof can be done by induction on the queries.

The only non trivial cases are those concerning rules (\textsc{TypeStar}) and (\textsc{TypeCount}). To evaluate these rules in polynomial time it suffices to recognize that a set $\psete$ can be interpreted as the set of edges in a schema element graph. Evaluating these rules, hence, is equivalent to the computation of the reflexive and transitive closure of the graph (bounded, in the case of rule (\textsc{TypeCount})). The unbounded closure can be computed in polynomial time by exploiting the Warshall's algorithm, while the bound closure  can be computed in polynomial time by relying on the usual squaring method.
\end{proof}

Basic type inference for RPQs is not only sound, but also complete.  Proof of completeness relies on a number of definitions and properties. 

The first definition specifies the set of paths matching a query.

\begin{definition}\label{def:paths}
Given a RPQ  $q$ on graphs over a finite alphabet $\Sigma$, the set of paths that can match $q$ is recursively defined as follows.
\begin{eqnarray*}
\paths{\epsilon} & = & \{\epsilon\}\\
\paths{a} & = & \{a\}\\
\paths{q_{1} + q_{2}} & = & \paths{q_{1}} \cup \paths{q_{2}}\\
\paths{q_{1} \cdot q_{2}} & = & \paths{q_{1}} \times \paths{q_{2}}\\
\paths{q^{*}} & = & \cup_{i \geqslant 0} \paths{q}^{i}
\end{eqnarray*}
\end{definition}

The second definition specifies when two nodes $u$ and $v$ are connected by a path $p$ in a graph $G$. A path can be either $\epsilon$ (the empty path) or a path $a \cdot p'$, where $a$ is an edge label, and $p'$ is  a path in turn. Path concatenation $p_1 \cdot p_2$ is defined in the obvious way.

%
%
%
%

\begin{definition}
Given a graph $G = (V, E, \rho)$, we say that two nodes $u, v \in V$ are connected by  a path $p$ if either $p=\epsilon$ and $u=v$, or  $p=a \cdot p'$ and there exists a pair $u,u' \in V $ such that $ (u, a, u') \in E $ and $p'$ connects nodes $u',v \in V$.

\end{definition}

It is easy to prove that, if $p_1$ connects $u,v$ and $p_2$ connects $v,t$ in $G$, then $p_1 \cdot p_2$ connects $u,t$ in $G$. The following lemma relates RPQ semantics, path semantics, and graph paths. 

\begin{lemma}\label{lemma:auxc0} Given a RPQ  $q$ and a graph $G = (V, E, \rho)$, if $u, v \in V$ are connected by  a path $p$ in $G$, and $p\in \paths{q}$, then $(u,v)\in \ggsem{q}$. 

\end{lemma}

\begin{proof}
By structural induction on $q$.

\begin{description}
\item[($q = \epsilon$)]  Trivial.

\item[($q = a$)] If  $p\in \paths{q}$, it must be that $p=a$. By definition of the semantics of RPQs, it follows that each pair of nodes $u, v \in V$ that is connected by  $p=a$ is in $\ggsem{q}$ .

\item[($q = q_1 + q_2$)] Simple induction.

\item[($q = q_1 \cdot q_2$)] Given that $p\in \paths{ q_1 \cdot q_2}$,  we have that $p=p_1 \cdot p_2$ with $p_i\in \paths{q}$, $i=1,2$. Also, since $p$ connects $(u,v)$,  $p_1$ connects  ($u, u_1$)  and $p_2$ connects ($u_1, v$) for a node $u_1$ in $G$. By induction we have $(u,u_1)\in \ggsem{q}$ and $(u_1,v)\in \ggsem{q}$, so the thesis follows by definition of  $\ggsem{q}$.

\item[($q= \reps{q_{1}}$)] Given that $p\in \paths{\reps{q_{1}}}$, we have that $p=p_1 \cdot \ldots \cdot p_n$  with $p_i\in \paths{q_1}$ and  $i=1\ldots n$. Since $p$ connects $(u, v) \in V$, we have that $(u, u_1), (u_1, u_2), \ldots, (u_{n-1}, u_n)$ in $G$ are respectively connected by $p_1, p2, \ldots, p_n$. So by induction we have $p_i \in \ggsem{q}$ and the thesis follows by definition of $\ggsem{q}$.
 \end{description}
\end{proof}

We need now to define paths over schemas.
\begin{definition}

Given a  schema $\schemas = \{\elemei\}_{i = 0}^{n}$,  we  say that two types   $e_o, e_t \in S$ are connected by  a path $p$ if either $p=\epsilon$ and $e_o = e_t$, or $p = a \cdot p'$ and there exists  $\eleme{}\in S$ such that:

\begin{itemize}
\item $a \in \sym{\gout{e_o}} \cap \sym{\gin{e}}$, and
\item   $e,  e_t$ are connected by  $p'$.
\end{itemize}
\end{definition}
The following lemma relates RPQ typing and paths over schemas.
\begin{lemma}\label{lemma:auxc1}

Given a schema $\schemas$ and a query $q \in \rpq$ whose inputs are described by $\schemas$, if $\judbasic{\schemas}{q}{\psete}$, then, for each $(\xeleme{i}, \xeleme{j}) \in \psete$ there exists a path $p$ over $S$ such that $p$ connects $(\xeleme{i}, \xeleme{j})$ and $p\in \paths{q} $.

\end{lemma}

\begin{proof}
By structural induction on $q$. Cases $q = \epsilon$,  $q = a$, and $q = q_1 + q_2$ are trivial. 
\begin{description}

\item[($q = q_1 \cdot q_2$)] We have that $\judbasic{\schemas}{q_{i}}{\psete_{i}}$, with $i=1,2$, and $\psete= {\psete_{1} \rconc \psete_{2}}$.   By induction we have that, for each pair $(\xeleme{1}, \xeleme{3})\in \psete_1$ and $(\xeleme{3}, \xeleme{2})\in \psete_2$,  there exist $p_1$ connecting $(\xeleme{1}, \xeleme{3})$ and $p_2$ connecting $(\xeleme{3}, \xeleme{1})$ on $\schemas$. So the thesis follows by taking $p=p_1 \cdot p_2$.

\item[($q=\reps{q_{1}}$)] The case is similar to the above once observed that $ (\xeleme{i}, \xeleme{j})\in \psete$ means that there exist  $(\xeleme{i}, \xeleme{1}')$, $(\xeleme{1}', \xeleme{2}') $, $\ldots$, $(\xeleme{n}', \xeleme{2})$ in $\psete$ and that each of these couples is connected by a path $p_i$ on $\schemas$ such that $p_i\in \paths{q_1}$, with $i=1,\ldots, n$.

 \end{description}
\end{proof}

To prove completeness of our RPQ typing rules we also need  the following definition and couple of lemmas. 

\begin{definition} Given a schema $\schemas$ and its double normalisation $\xxnorm{\schemas}$, we indicate with $F_S$ the function from $\xxnorm{\schemas}$ to $\schemas$ associating to each type $\eleme$ in  $\xxnorm{\schemas}$ the unique type $\eleme'$ in $\schemas$ from which $\eleme$ has been generated (Definition \ref{def:dnorm}).

Also, for every $\eleme'$ in ${\schemas}$, $F^{-1}_{S}(\eleme')$ is the set of element types generated by $\eleme'$ by means of double normalisation. 

\end{definition}

\begin{lemma}\label{lem:auxc4} Given a schema $\schemas$ with the pair of element types $(\eleme_1, \eleme_2 )$, if this pair is connected by $p$ on $\schemas$, then there exist $\eleme'_i\in F^{-1}_{S}(\eleme_i)$ for $i=1,2$ such that $(\eleme'_1, \eleme'_2 )$ are connected by $p$ in $\xxnorm{\schemas}$.
\end{lemma}

\begin{proof} By induction on $|p|$. If $p=a$, then there exist two types $e_1$ and $e_2$ in $\schemas$ such that $a\in \sym{{e_1}.out }\cap \sym{{e_2}.in }$. We have $\xnorm{\sym{{e_1}.out}}=C_1 + \ldots + C_n$ and $a\in \sym{C_i}$ for $i=1,\ldots,n$. Similarly, $\xnorm{\sym{{e_2}.in}}=C'_1 + \ldots + C'_m$ and $a\in \sym{C'_j }$ for for $j=1,\ldots,m$. By definition of double normalisation we know that there will be $e'_1$ and $e'_2$  obtained by normalisation of $e_1$ and $e_2$  such that ${e'_1}.out = C_i$ and ${e'_2}.in=C'_j$, which prove the basic case. 

Concerning the  case  $p=a.p'$ with $p'=b.p''$ (if $p'$ is empty the case has already been proved), by hypothesis we have that $(\eleme_1, \eleme_3 )$ is connected by $a$ and that $(\eleme_3, \eleme_2 )$ is connected by $p'$. By induction we can assume that there exists in $\xxnorm{\schemas}$ a couple  $(\eleme'_1, \eleme'_3 )$ connected by $a$ and  a couple $(\eleme''_3, \eleme'_2 )$ connected by $p'$, with $\eleme'_3$ and $\eleme''_3$ in $F^{-1}_{S}(\eleme_3)$. It is easy to prove that the type $\eleme'''_3=({\eleme'_3}.in, {\eleme''_3}.out)$  is in  $F^{-1}_{S}(\eleme_3)$ and that it both receives $a$ and emits $b$. So we have that $a$ connects $(\eleme'_1, \eleme'''_3 )$ and $p'$ connects $(\eleme'''_3, \eleme'_2 )$ in $\xxnorm{\schemas}$, so the thesis is proved.

\end{proof}

\begin{lemma}\label{lemma:auxc3} For any  well formed schema  $\schemas$   there exists a graph $G$ that respects  $\schemas$ and such that: if $(\eleme_1, \eleme_2 )$ are connected by $p$ in ${\schemas}$, then two nodes $(n_1,n_2)$ in $G$ are connected by $p$, with $n_i$ having type $\eleme_i$, for $i=1,2$. 

\end{lemma}

\begin{proof}
For any $(\eleme_1, \eleme_2 )$ in $\schemas$ by Lemma \ref{lem:auxc4} we have that there exist $\eleme'_i\in F^{-1}_{S}(\eleme_i)$ for $i=1,2$ such that $(\eleme'_1, \eleme'_2 )$ are connected by $p$ in $\xxnorm{\schemas}$. We then prove that the desired $G$ exists for $\xxnorm{\schemas}$ with nodes $(n_1,n_2)$ having types $(\eleme'_1, \eleme'_2 )$, and therefore $(\eleme_1, \eleme_2 )$.

The proof of the existence of such a $G$ is quite similar to that for Theorem \ref{theo:rpqcompl}, so we omit  the details, and just observe that  the graph $G$ we are looking for is actually the graph built in that proof, which ensures that each node of the built graph corresponds to a different type in $\xxnorm{\schemas}$, and that each node of a type $\eleme'$ in $\xxnorm{\schemas}$ has  a incoming (outgoing) $a$-edge for each $a\in \sym{\eleme'.in}$ (for each $a\in \sym{\eleme'.out}$) . This directly implies that we can connect  the two nodes $(n_1, n_2)$ in $G$ corresponding to types $(\eleme'_1, \eleme'_2 )$ with the path $p$. 
\end{proof}

We have now all the tools to prove completeness of RPQ typing over well-formed schemas. 

\begin{theorem}\label{theo:rpqcompl}
Given a well-formed schema $\schemas$ and a query $q \in \rpq$ whose inputs are described by $\schemas$, if $\judbasic{\schemas}{q}{\psete}$, then, for each $(\xeleme{i}, \xeleme{j}) \in \psete$, there exists $G \in \ssem{\schemas}$ such that $\ggsem{q}$ contains $(u,v)$, where $u \in \ssem{\xeleme{i}}$ and $v \in \ssem{\xeleme{j}}$.
\end{theorem}

\begin{proof}
By Lemma \ref{lemma:auxc1} we have that $(\xeleme{i}, \xeleme{j})$ are connected by a path $p$ in $\schemas$ and that $p\in \paths{q}$. By Lemma  \ref{lemma:auxc3} we have that there exists a graph $G$ meeting $\schemas$ an including two nodes $(u,v)$ connected by $p$. So by Lemma \ref{lemma:auxc0}  we can conclude that  $(u,v)\in \ggsem{q} $ and $u \in \ssem{\xeleme{i}}$ and $v \in \ssem{\xeleme{j}}$.

\end{proof}

Theorem \ref{theo:rpqcompl} has important consequences. Indeed, as shown by Corollary \ref{cor:sat}, we can prove that, for RPQs, the following satisfiability problem can be decided in polynomial time.

\DC{SAT deve avere anche S come parametro?}
\begin{problem}[$\xsat{\chi}$]
Given a   well-formed  $\schemas$ and a query $q$ in a language $\chi$, is there a graph $G \in \ssem{\schemas}$ such that $\ggsem{q} \neq \emptyset$?
\end{problem}


\begin{corollary}\label{cor:sat}
$\rsat$ can be decided in polynomial time.
\end{corollary}

\begin{proof}
Consider a query $q$ and a well-formed $\schemas$.  We must first prove that:
$$(\forall G \in \ssem{\schemas}: \  \ggsem{q} = \emptyset) \iff  \judbasic{\schemas}{q}{\emptyset}$$
\begin{description}
\item[$(\Rightarrow)$]
 Assume that there is no graph $G \in \ssem{\schemas}$ such that $\ggsem{q} \neq \emptyset$. Then, by Theorem  \ref{theo:rpqcompl}, $\judbasic{\schemas}{q}{\emptyset}$. Indeed, if $\judbasic{\schemas}{q}{\psete}$ with $\psete \neq \emptyset$, by Theorem \ref{theo:rpqcompl} there would exist at least one graph $G \in \schemas$ for which $\ggsem{q} \neq \emptyset$, which is a contradiction.

\item[$(\Leftarrow)$]  
Assume that  $\judbasic{\schemas}{q}{\emptyset}$. Then, by Theorem \ref{theo:basicsound}, there is no graph $G \in \ssem{\schemas}$ such that $\ggsem{q} \neq \emptyset$. Indeed, if  $\ggsem{q} \neq \emptyset$, then, by Theorem \ref{theo:basicsound}, $\judbasic{\schemas}{q}{\psete}$, where $\psete \neq \emptyset$, which is a contradiction.
\end{description}

The fact that $\judbasic{\schemas}{q}{\psete}$ can be evaluated in polynomial (by Theorem \ref{theo:bcomplx}) completes the proof.
\end{proof}


\medskip

While sound, basic type inference is not complete for NREs and GXPath queries, as shown in the following example. 

\begin{example}
Consider the following graph schema:
\[
\begin{array}{lllllllllll}
\xeleme{1} & = & (\epsilon, \ls{a}{1} + \ls{b}{1}) && \xeleme{4} & = & (\reps{(\ls{c}{1})}, \epsilon)\\
\xeleme{2} & = & (\reps{(\ls{a}{1})}, \ls{c}{1}) && \xeleme{5} & = & (\reps{(\ls{d}{1})}, \epsilon)\\
\xeleme{3} & = & (\reps{(\ls{b}{1})}, \ls{d}{1})&& 
\end{array}
\]

Consider now the query $q = [b] \cdot a \cdot c$. This query looks for nodes having outgoing edges labelled with $a$ and $b$. As $\xeleme{1}$ prescribes that these edges are mutually exclusive, the result of this query is always empty. However, the rules first infer the set $\{(\xeleme{1}, \xeleme{1})\}$ for the nested regular expression, and the set $\{(\xeleme{1}, \xeleme{4})\}$ for the $a \cdot c$. The inferred set, hence, is $\{(\xeleme{1}, \xeleme{4})\}$.
\end{example}

\section{Related Works}\label{sec:relworks}
Describing the structure of graphs is a subject that has been analyzed only in a few papers. Graph grammars \cite{DBLP:conf/ijcai/PfaltzR69} probably represent the most widely known approach for describing graphs. As a plain string grammar, a graph grammar shows how a graph can be generated starting from a source node, by applying a set of production rules. As for string or tree grammars, graph grammars are used for generating graphs, for transforming existing graphs into new ones, or for pattern matching, but they are not suitable for type inference.

TSL is the schema language of Trinity \cite{DBLP:conf/sigmod/ShaoWL13}, a main-memory graph processing system based on the Microsoft ecosystem. By using a TSL script, which is compiled in .NET object code,  it is possible to specify the structure of nodes, which can have richly defined values, e.g., those required by BFS and DFS algorithms, as well as the type of outgoing edges; apparently, there is no way to describe constraints on incoming edges, which can have any cardinality.

SheX \cite{DBLP:conf/icdt/StaworkoBGHPS15} is a schema language for RDF data. As in TSL, in SheX it is possible to describe complex node structures, and, unlike in TSL, outgoing edges can be defined by using regular expressions. However, just as in TSL, there is no way to specify constraints on incoming edges. This means that, for instance, in a schema describing cars and car owners, one can impose the constraint that a single person can own at most $n$ cars, but not the constraint that a car can have one single owner at a time. This makes impossible to define empty SheX schemas, but it limits the expressivity of the language.


\section{Conclusions and Future Work}\label{sec:concl}

In this paper we described a schema language for data graphs, introduced a mild restriction that makes impossible to define empty schemas, and used this schema formalism to create a type inference system for RPQs, NREs, and GXPath queries. This type inference system is sound and, in the case of RPQs, it is also complete, hence making possible to check the satisfiability of a query just by looking at its inferred type.

In the near future we plan to work on three directions. First, we want to understand if schema well-formedness can be checked in polynomial or if double normalization is mandatory. Second, we want to better investigate the emptiness problem for a more general class of graph schemas and try to relax the 1/* constraint, in particular by exploring approaches for checking the consistency of systems of linear diophantine equations with linear parameters. Finally, we want to study type inference techniques that return more detailed information about an input query, e.g., the set of paths that the query may traverse in an input graph matching a given schema. 


\section{Acknowledgements}\label{sec:ack}

Authors would like to thank Maria Grazia Russo for her suggestions about systems of diophantine equations.


\end{document}